\documentclass[journal]{IEEEtran}
\usepackage[utf8]{inputenc}
\usepackage{amsmath,amssymb,amsfonts}
\usepackage{graphicx, color,soul}
\usepackage{booktabs,float}
\usepackage{cite}
\usepackage{algorithm}
\usepackage{algpseudocode}
\usepackage{hyperref}

\newtheorem{remark}{Remark}
\newtheorem{theorem}{Theorem}

\newtheorem{proof}{Proof}
\newtheorem{assumption}{{\bf{Assumption}}}
\newtheorem{problem}{{\bf{Problem}}}

\newcommand\xqed[1]{%
  \leavevmode\unskip\penalty9999 \hbox{}\nobreak\hfill
  \quad\hbox{#1}}
\newcommand\demo{\xqed{$\triangle$}}

\title{Data-Driven Switched State Estimation with Sparse Sensor Scheduling for Nonlinear Networked Systems}
\author{Gianfranco Gagliardi, Franco Angelo Torchiaro, Vincenzo Gallelli, Ayman El Qemmah and Alessandro Casavola
\thanks{G. Gagliardi, F.A. Torchiaro, A. El Qemmah and A. Casavola are with the Department of Computer, Modeling, Electronics and Systems Engineering (DIMES), University of Calabria. \\ 
V. Gallelli is with the Department of Civil Engineering (DINCI), University of Calabria. }}

\begin{document}

\maketitle

\begin{abstract}
This paper presents a data-driven framework for joint observer design and sparse sensor scheduling for unknown nonlinear networked systems. The nonlinear dynamics are approximated online as a piecewise sequence of locally linearized discrete-time systems, resulting in a switched linear representation recursively identified through Subspace State-Space System Identification (4SID). To ensure consistency across regime transitions, an Orthogonal Procrustes alignment is introduced to promote coordinate consistency across consecutive regime transitions and mitigate artificial discontinuities caused by arbitrary state-space coordinate changes. Based on the identified local realizations, a dual-rate predictor–corrector observer is designed. The observer gain is computed through a convex optimization problem that jointly addresses estimation accuracy, sensor sparsity, and stability requirements. In particular, an $L_{2,1}$-norm regularization term promotes column sparsity in the gain matrix, enabling the automatic selection of informative measurement channels, while a spectral-norm constraint guarantees Schur stability of the estimation error dynamics. The proposed framework is first validated on a traffic network simulated in Aimsun Next. Results obtained on an $18$-link network show that the observer accurately reconstructs macroscopic traffic states while significantly reducing the number of active sensors required for real-time estimation.
\end{abstract}

\begin{IEEEkeywords}
Data-driven estimation, dynamic sensor selection, subspace identification, convex optimization, cyber-physical systems, macroscopic traffic tracking.
\end{IEEEkeywords}
\section{Introduction}
\label{s1}
Optimal sensor placement in large-scale dynamical systems is essential for reducing the number of active sensing units
while maintaining accurate monitoring and state estimation \cite{Joshi, Polyak, Roy}. Addressing this challenge, this paper presents a framework for concurrent dynamic sensor selection and observer design to jointly improve estimation accuracy and reduce sensing requirements in large-scale nonlinear systems.
%
Extensive research has addressed the problem of optimal sensor selection and placement in dynamical networks. Existing model-based approaches can generally be classified into three main categories. The first includes combinatorial and graph-based methods that exploit the structural properties of the underlying network \cite{Tzoumas1}, \cite{Tzoumas2}, \cite{Gagliardi}. A second class relies on mixed-integer optimization and convex relaxation techniques to determine sparse or minimum-cardinality sensor configurations \cite{Chanekar}, \cite{Taha}. A third approach formulates the problem within a semidefinite programming (SDP) framework, often incorporating sparsity-promoting regularization terms to jointly optimize estimation performance and sensor usage \cite{Dhingra}, \cite{Lin}, \cite{GagliardiLetter}.
Among these regularization strategies, the $L_{2,1}$-norm (Group Lasso \cite{yuan2006model}) is particularly effective for structured sensor selection. Unlike the standard $L_1$ norm, which promotes element-wise sparsity, the $L_{2,1}$ norm enforces sparsity at the group level by penalizing the Euclidean norm of entire variable blocks. In observer design, this property naturally promotes column sparsity in the observer gain matrix, allowing uninformative measurement channels to be completely discarded \cite{joshi2009sensor}.
Despite their effectiveness, these approaches generally rely on the availability of an accurate analytical state-space model of the underlying physical system. In modern large-scale cyber-physical systems (e.g. transportation networks, power systems, and large-scale biochemical processes), however, deriving such models from first principles is often impractical.


As a result, data-driven approaches have attracted increasing attention in recent years. Rather than relying on analytical models, these methods exploit operational input-output data to identify and adapt models of the system online. Within this context, Deep Neural Networks (DNNs), Recurrent Neural Networks (RNNs), and, more recently, Physics-Informed Neural Networks (PINNs) have been widely adopted to approximate nonlinear dynamics, estimate unmeasurable states, and construct black-box or hybrid observers \cite{Lusch2018, Raissi2019, Revach2022}. Despite their strong approximation capabilities, neural-network-based approaches often require large training datasets, provide limited interpretability, and generally lack explicit stability guarantees for the resulting estimation dynamics.

To bridge the gap between model-free adaptation and rigorous stability guarantees, recent advances in control theory have increasingly leveraged the behavioral framework based on Willems’ Fundamental Lemma \cite{willems2005note}. Building on this perspective, data-driven methodologies have been extended to the synthesis of stabilizing observers directly from input-output Hankel matrices, without requiring explicit model identification \cite{Wolff2024}. In parallel, Subspace State-Space System Identification (4SID) methods \cite{vanoverschee1996subspace, katayama2006subspace} provide an effective algebraic framework for extracting local linear time-invariant (LTI) realizations directly from measured data.
However, despite these advances, most existing data-driven estimation frameworks assume a static and fully available sensor configuration. 
Motivated by the need to combine adaptive state estimation with dynamic sensor selection, this paper proposes a fully data-driven framework for nonlinear networked systems. The underlying nonlinear process is approximated online through a switched linear representation, where local models are recursively identified from streaming input-output data. Based on these local realizations, an adaptive sparse observer is synthesized to jointly guarantee estimation accuracy, stability, and sensor reduction.

The main contributions of this work are summarized as follows:
\begin{enumerate}
\item Data-Driven Switched Approximation: Nonlinear dynamics are modeled online as switched linear systems via a recursive Subspace State-Space System Identification (4SID) procedure, avoiding explicit first-principles modeling and adapting to regime shifts.
\item Coherent Subspace Alignment: An Orthogonal Procrustes alignment preserves coordinate consistency across regime transitions, preventing artificial state trajectory discontinuities.
\item Joint Observer Synthesis and Sensor Sparsification: A predictor-corrector observer is designed using the identified local realizations. The observer gains are computed via convex optimization with an $L_{2,1}$-norm penalty promoting column sparsity in the gain matrix and automatically discarding uninformative sensor channels.
\item Stability-Constrained Estimation: A spectral norm constraint ensures Schur stability of local error dynamics, guaranteeing nominal asymptotic stability and uniform ultimate boundedness under disturbances.  
\end{enumerate}

The proposed framework is validated on a traffic network simulated in Aimsun Next, where the observer accurately reconstructs the traffic states while significantly reducing the number of active sensors required for real-time estimation.


The paper is organized as follows: Section \ref{sec:preliminaries} provides mathematical preliminaries, and Section \ref{sec:sys_prob} formalizes the estimation and scheduling problem. Sections \ref{sec:proposed_solution} and \ref{sec:opt_formulation} detail the data-driven architecture, including the coherent subspace identification, Procrustes alignment, and convex optimization framework. Section \ref{sec:properties} establishes theoretical stability guarantees. Finally, Section \ref{sec:simulation} presents the Aimsun Next simulation results, followed by conclusions in Section \ref{sec:conclusion}.


\section{Preliminaries}
\label{sec:preliminaries}
Before detailing the proposed methodology, in this section the notation used throughout the paper is introduced. The sets of real numbers, non-negative integers, and positive integers are denoted by $\mathbb{R}$, $\mathbb{Z}_{\ge 0}$, and $\mathbb{Z}_{> 0}$, respectively. For a matrix $M \in \mathbb{R}^{n \times m}$, $M^T$ denotes its transpose and $M^\dagger$ denotes its Moore-Penrose pseudo-inverse. The induced 2-norm (spectral norm) is denoted by $\|M\|_2$, the Frobenius norm by $\|M\|_F$, and the mixed $L_{2,1}$-norm (Group Lasso) by $\|M\|_{2,1} = \sum_{i=1}^m \|M^{(i)}\|_2$, where $M^{(i)}$ denotes the $i$-th column of $M$. Furthermore, the mathematical expectation operator is denoted by $\mathbb{E}[\cdot]$. 

Consider two consecutive identified local realizations $\Sigma_j = (\hat A_j,\hat B_j,\hat C_j)$
and
$\Sigma_{j+1} = (\hat A_{j+1},\hat B_{j+1},\hat C_{j+1})$
associated with adjacent operational intervals $I_j $and $I_{j+1}$.
The transition between $\Sigma_j$ and $\Sigma_{j+1}$ is said to be coordinate-coherent if there exists a transformation
\begin{equation}\label{e1}
T_j \in O(n):= \{ T \in \mathbb{R}^{n \times n} \mid T^T T = I \}
\end{equation}
such that the updated realization satisfies
$\hat A_{j+1} \approx T_j^\top \hat A_j T_j,$ $\hat B_{j+1} \approx T_j^\top \hat B_j$ and
$\hat C_{j+1} \approx \hat C_j T_j$, while preserving the Euclidean norm
$\|T_j x\|_2 = \|x\|_2, \forall x \in \mathbb R^n$.
Consequently, consecutive state-space realizations evolve through norm-preserving coordinate transformations, preventing artificial amplification of the observer error during regime transitions.

In data-driven control and system identification, a fundamental tool is the block Hankel matrix. Given a discrete-time vector sequence $s = \{s_1, s_2, \dots, s_T\}$ with $s_k \in \mathbb{R}^d$, a user-defined past horizon $L$, and a window length $N$, the block Hankel matrix of depth $L$ is constructed as
\begin{equation}
	H_{L,N}(s) = 
	\begin{bmatrix}
		s_1 & s_2 & \dots & s_N \\
		s_2 & s_3 & \dots & s_{N+1} \\
		\vdots & \vdots & \ddots & \vdots \\
		s_L & s_{L+1} & \dots & s_{L+N-1}
	\end{bmatrix} \in \mathbb{R}^{d L \times N}.
	\label{eq:generic_hankel}
\end{equation}

\section{System Description and Problem Formulation}
\label{sec:sys_prob}

\subsection{Nonlinear System Description}
Consider a dynamic physical process governed by an unknown nonlinear discrete-time system:
\begin{subequations}
	\label{eq:nlsys}
	\begin{align}
		x_{k+1} &= f(x_k, u_k) + w_k, \label{nlstate}\\
		y_k &= h(x_k) + v_k, \label{nloutput}
	\end{align}
\end{subequations}
where:
\begin{itemize}
	\item $x_k \in \mathbb{R}^n$ is the unmeasured physical state of the system with initial condition $x_0 \in \mathbb{R}^n$;
	\item $u_k \in \mathbb{R}^{n_u}$ is the known input vector;
	\item $y_k \in \mathbb{R}^{n_y}$ is the measured output vector collected from the complete set of available sensors;
	\item $f(\cdot)$ and $h(\cdot)$ represent the unknown nonlinear state transition and observation functions;
	\item $w_k$ and $v_k$ denote process and measurement disturbance sequences. For the theoretical analysis, disturbances are assumed to be deterministically bounded. In the numerical simulations, Gaussian disturbances are used as a standard stochastic
benchmark; the corresponding boundedness statements should therefore be interpreted in a high-probability sense.
\end{itemize}

\subsection{Problem Statement}

This paper focuses on the design of an adaptive architecture that concurrently provides:
\begin{enumerate}
	\item[$i$)] \textit{Dynamic Sensor Scheduling:} A policy $\gamma_k$ that minimizes the number of active sensors (i.e. minimizes $\|\gamma_k\|_0$) while retaining sufficient information for accurate estimation.
	\item[$ii$)] \textit{State Estimation:} An estimator generating an approximation $\hat{x}_k$ of the true unmeasured state $x_k$ using solely the known inputs $u_k$ and the sparse subset of active measurements defined by $\gamma_k$.
	\item[$iii$)] \textit{Stability Guarantees:} A theoretical guarantee that the resulting estimation error $e_k = x_k - \hat{x}_k$ remains uniformly bounded and converges asymptotically despite structural shifts in the nonlinear dynamics or the sensor schedule.
\end{enumerate}

Specifically, the objective is to address this problem in the case where the explicit analytical forms of the nonlinear system functions $f(\cdot)$ and $h(\cdot)$ are \textit{unknown}, but offline and online input-output data samples from the nonlinear system are available. More in detail, the following assumption is made 

\begin{assumption}
	\label{ass:unknowndynpersistentdata}
	Suppose that the explicit analytical forms of the nonlinear system functions in \eqref{nlstate}--\eqref{nloutput} are \textit{unknown}, but input-output data $\{u_k, y_k\}$ are available, where the input sequence $u_k$ is persistently exciting, i.e. sufficiently rich to excite all relevant dynamic modes of the system \cite{vanoverschee1996subspace}. Furthermore, the system state dimension $n$ is known or can be estimated offline from historical data.
\end{assumption}

Then, the problem to be solved in the data-driven setting can be stated as follows:

\begin{problem}
	\label{problemddsess}
	\emph{Under Assumption \ref{ass:unknowndynpersistentdata}, design a data-driven state estimation and dynamic sensor selection framework. Specifically, the estimator must concurrently reconstruct the unmeasured
state trajectories \(x_k\) and select a sparse and informative subset of active sensors.}
\end{problem}

\section{Data-Driven Switched Architecture}
\label{sec:proposed_solution}

To solve Problem \ref{problemddsess} without resorting to analytical linearization, 
we propose a framework that sequentially approximates the nonlinear system using localized data-driven models. 

\subsection{Switched Linear Approximation}

Since the analytical forms of $f(\cdot)$ and $h(\cdot)$ are unknown and cannot be readily derived from first principles, a single global model is inadequate to describe the system behavior. Under the assumption that the underlying process evolves over a finite set of distinct operating regimes, the overall dynamics can be represented as a discrete-time switched system \cite{liberzon2003switching}.

To formalize this framework, let $\mathcal{K} = \{k_0, k_1, k_2, \dots \}$ denote a monotonically increasing sequence of switching instants, where each integer $k_j$ represents the specific discrete time step at which the $j$-th regime transition occurs, with the initial time $k_0 = 0$. This sequence partitions the semi-infinite discrete temporal horizon $\mathbb{Z}_{\ge 0}$ into a set of contiguous, non-overlapping operational intervals
\begin{equation}
	\mathcal{I}_j = [k_j, k_{j+1}), \quad \text{such that} \quad \bigcup_{j=0}^{\infty} \mathcal{I}_j = \mathbb{Z}_{\ge 0}.
\end{equation}
In such a framework, the temporal length of each interval, representing the operational dwell-time $\tau_j = k_{j+1} - k_j$, is governed by a hybrid transition logic. The boundary $k_{j+1}$ is dynamically determined either by reaching a predefined periodic scheduling threshold $T_{\rm sched} \in \mathbb{Z}_{>0}$, or by the asynchronous activation of an event-detection indicator $\mathcal{E}(k) \in \{0, 1\}$ (e.g., an anomaly tracking flag) \cite{tabuada2007event}, \cite{heemels2012introduction}. Mathematically, the next switching instant is defined as

\begin{equation}
\begin{aligned}
k_{j+1}
=
\min \Big\{ k>k_j \;:\;&
(k-k_j\geq T_{\min}) \\
&\wedge
\big((k-k_j\geq T_{\rm sched})
\vee \mathcal{E}(k)=1\big)
\Big\}.
\end{aligned}
\end{equation}
If an event is detected before the minimum dwell-time is
satisfied, the event is stored and the adaptation is delayed
until \(k-k_j\geq T_{\min}\).

Within each operational interval $k \in \mathcal{I}_j$, the local dynamics are approximated by a piecewise linear time-invariant (LTI) representation \cite{rugh1991analytical}
\begin{subequations}
	\begin{align}
		x_{k+1} &\approx A_j x_k + B_j u_k + w_k, \\
		y_k &\approx C_j x_k + v_k,
	\end{align}
	\label{eqSS}
\end{subequations}
where $A_j$, $B_j$, and $C_j$ are the equivalent linear system matrices governing the $j$-th regime. In this context, the stochastic sequences $w_k$ and $v_k$ are used to account for bounded residual errors arising from the local linear approximation of the underlying nonlinear vector fields.

\subsection{Coherent Subspace Identification}
Instead of computing these matrices via analytical linearization, the proposed approach relies on a data-driven identification procedure. Specifically, upon detection of a regime transition at time $k_j$, a local historical window of input–output data pairs $\{u_k, y_k\}$ is used to sequentially identify the corresponding realization $(\hat{A}_j, \hat{B}_j, \hat{C}_j)$ via Coherent Subspace Identification.

To identify the regime-dependent linear matrices $(\hat{A}_j, \hat{B}_j, \hat{C}_j)$, the recursive Subspace State-Space System Identification (4SID) approach of \cite{vanoverschee1996subspace} is employed over a sliding historical data batch of length $W_{id}$. 
In particular, upon the activation of a new interval $\mathcal{I}_j$, the most recent input–output data pairs are arranged into block Hankel matrices. To this end, let $L_h$ denote a user-defined past horizon, with $L_h > n$ and
\begin{equation}
	N = W_{id} - L_h + 1.
	\label{eq00}
\end{equation}
The output block-Hankel matrix $Y_H \in \mathbb{R}^{n_y L_h \times N}$ is obtained as defined in \eqref{eq:generic_hankel}, i.e. 
\begin{equation}
	Y_H = 
	\begin{bmatrix}
		y_1 & y_2 & \dots & y_N \\
		y_2 & y_3 & \dots & y_{N+1} \\
		\vdots & \vdots & \ddots & \vdots \\
		y_{L_h} & y_{L_h+1} & \dots & y_{W_{id}}
	\end{bmatrix}.
\end{equation}
Analogously, the input block-Hankel matrix $U_H \in \mathbb{R}^{n_u L_h \times N}$ is obtained using the input sequence $u_k$.

Although the global physical process in \eqref{nlstate}–\eqref{nloutput} is inherently nonlinear, its behavior within any isolated operational interval $\mathcal{I}_j$ can be approximated by a local LTI realization. Under Assumption \ref{ass:unknowndynpersistentdata}, the input sequence is sufficiently rich to excite all relevant local dynamic modes. In the online implementation, input richness is assessed numerically by checking whether
\begin{equation}
	\operatorname{rank}(U_H) = L_h \cdot n_u. \label{fullrankcond}
\end{equation}
If the rank condition is not satisfied, the adaptation step is
skipped and the current realization is retained.

This rank condition is a fundamental mathematical prerequisite in data-driven systems theory (acting, for instance, as the core hypothesis of Willems' Fundamental Lemma \cite{willems2005note}). In the context of subspace identification, condition \eqref{fullrankcond} guarantees that the input data are sufficiently rich to decouple the forced response from the autonomous dynamics, enabling the consistent algebraic extraction of the local state-space realization.
Moreover, according to standard results in subspace identification theory \cite{katayama2006subspace}, the recursive substitution of the linear state-space equations yields the following fundamental \textit{Data Equation}:
\begin{equation} \label{eq:data_equation}
	Y_H = \Gamma_j X_{k_j} + \mathcal{T}_j U_H + E_H,
\end{equation}
where $X_{k_j}$ is the state sequence matrix starting at the identification instant $k_j$, and $E_H$ is the aggregate perturbation Hankel matrix encapsulating both the stochastic noise sequences ($w_k, v_k$) and the unmodeled nonlinear residual dynamics. The extended observability matrix $\Gamma_j \in \mathbb{R}^{n_y L_h \times n}$ governs the autonomous response mapping, defined as
\begin{equation}
	\Gamma_j = 
	\begin{bmatrix}
		C_j \\
		C_j A_j \\
		\vdots \\
		C_j A_j^{L_h-1}
	\end{bmatrix},
\end{equation}
and $\mathcal{T}_j \in \mathbb{R}^{n_y L_h \times n_u L_h}$ is the lower block-triangular Toeplitz matrix containing the system Markov parameters, representing the forced response. Since the strictly causal physical model implies no direct feedthrough ($D_j = 0$), it is structured as
\begin{equation}
	\mathcal{T}_j = 
	\begin{bmatrix}
		0 & 0 & \dots & 0 \\
		C_j B_j & 0 & \dots & 0 \\
		\vdots & \ddots & \ddots & \vdots \\
		C_j A_j^{L_h-2} B_j & \dots & C_j B_j & 0
	\end{bmatrix}.
\end{equation}
Equation \eqref{eq:data_equation} highlights that the output data are a linear combination of the autonomous response ($\Gamma_j X_{k_j}$) and the forced response ($\mathcal{T}_j U_H$). 

To isolate the subspace spanned by $\Gamma_j$, the influence of the input term $\mathcal{T}_j U_H$ must be nullified. This is achieved by projecting the row space of $Y_H$ onto the orthogonal complement of the row space of $U_H$. The corresponding orthogonal projection operator is defined as:
\begin{equation} \label{eq:projection_matrix}
	\Pi_{U^\perp} = I - U_H^T (U_H U_H^T)^{\dagger} U_H.
\end{equation}
By construction, post-multiplying the input Hankel matrix by this operator yields the zero matrix, i.e., $U_H \Pi_{U^\perp} = \mathbf{0}$. 

In order to isolate the autonomous dynamics from the deterministic forced response, this right-multiplication is applied to the fundamental data equation \eqref{eq:data_equation}, i.e. 
\begin{align} \label{eq:projection_derivation}
	Y_H \Pi_{U^\perp} &= (\Gamma_j X_{k_j} + \mathcal{T}_j U_H + E_H) \Pi_{U^\perp} \nonumber \\
	&= \Gamma_j X_{k_j} \Pi_{U^\perp} + \mathcal{T}_j (U_H \Pi_{U^\perp}) + E_H \Pi_{U^\perp} \nonumber \\
	&= \Gamma_j X_{k_j} \Pi_{U^\perp} + E_H \Pi_{U^\perp}.
\end{align}
Thus, the projected output matrix can be defined as
\begin{equation}
	Y_{clean} = Y_H \Pi_{U^\perp}.
	\label{eq:yclean}
\end{equation}
In an ideal deterministic setting ($E_H = \mathbf{0}$), the column space of $Y_{clean}$ 
is contained in the column space of the extended observability matrix $\Gamma_j$,
and its rank is exactly equal to the system order $n$. However, due to the presence of the perturbation term $E_H \Pi_{U^\perp}$, the matrix becomes full rank. 

Under Assumption \ref{ass:unknowndynpersistentdata}, a Truncated Singular Value Decomposition (TSVD) of rank $n$ is applied to the projected data to robustly extract the underlying unforced system realization while explicitly filtering out the unstructured noise and nonlinear residuals. According to standard subspace identification theory \cite{vanoverschee1996subspace, katayama2006subspace}, truncating the small singular values mathematically separates the dominant dynamic subspace from the noise subspace. Assuming the data window $N$ is sufficiently large, the SVD yields
\begin{align}
	Y_{clean} & = 
	\begin{bmatrix} W_{new} & W_{noise} \end{bmatrix} 
	\begin{bmatrix} \Sigma_{new} & 0 \\ 0 & \Sigma_{noise} \end{bmatrix} 
	\begin{bmatrix} V_{new}^T \\ V_{noise}^T \end{bmatrix} \nonumber \\
	&  \approx W_{new} \Sigma_{new} V_{new}^T, \label{eq:svd}
\end{align}
where $\Sigma_{new} \in \mathbb{R}^{n \times n}$ contains the $n$ dominant singular values corresponding to the physical system modes, while the sub-dominant singular values in $\Sigma_{noise}$ associated with the perturbation $E_H$ are structurally discarded.

In the context of subspace identification, the extracted state-space realization is only defined up to a similarity transformation. As a result, independently identified local models may be expressed in different coordinate bases, generating artificial discontinuities in the estimated system trajectory during regime transitions. To preserve consistency and achieve a \textit{bumpless transfer} across consecutive regime transitions, the newly identified subspace is aligned with the previous one through an Orthogonal Procrustes transformation \cite{schonemann1966generalized}, \cite{golub2013matrix}.  Let $W_{new}$ and $W_{old}$ denote the current and previous observability bases respectively. Then, the alignment matrix is obtained by seeking the optimal orthogonal transformation matrix $T_{align} \in \mathcal{O}(n)$ that minimizes the Frobenius distance between the two subspaces
\begin{equation}
	T_{align} = \arg\min_{T^T T = I} \| W_{new} T - W_{old} \|_F,
	\label{eq:align}
\end{equation}
with $\mathcal{O}(n)$ in (\ref{e1}) denoting the orthogonal group of matrices, whose elements preserve the Euclidean norm.

The exact analytical solution to this orthogonal Procrustes problem is obtained using a well-established algebraic sequence. First, the cross-correlation matrix $M \in \mathbb{R}^{n \times n}$, which quantifies the mutual geometric overlap between the newly extracted and historical subspaces, is constructed as
\begin{equation}
	M = W_{new}^T W_{old},
\end{equation}
and then its Singular Value Decomposition (SVD) is computed, i.e. 
\begin{equation}
	M = W_r \Sigma_r V_r^T.
\end{equation}
The optimal orthogonal transformation that minimizes the Procrustes distance is then synthesized by bypassing the singular value matrix $\Sigma_r$ and evaluating the direct product of the left and right singular vectors
\begin{equation}
	T_{align} = W_r V_r^T. \label{eq:closed_form}
\end{equation}
Once the aligned coherent basis 
\begin{equation}
	W_{coher} = W_{new} T_{align}
	\label{eq:coher}
\end{equation} 
is established, the internal state trajectories must be reconstructed to explicitly capture the forced dynamic response - a necessary condition for consistently estimating the input matrix $\hat{B}_j$. To achieve this, the raw measurement matrix $Y_H$ is orthogonally projected onto the coherent observability space, yielding the baseline state sequence $X_n \in \mathbb{R}^{n \times N}$ for the active operational window
\begin{equation}
	X_n = W_{coher}^{\dagger} Y_H.
\end{equation}
Although this direct pseudo-inversion inherently retains a residual coupling from the forced response dynamics (i.e., the unknown Markov parameters), it provides a computationally efficient approximation of the state sequence. Any induced numerical artefacts are systematically mitigated  via structural regularization during the subsequent least-squares extraction of the LTI system matrices.

To characterize the system transition dynamics, the data sequence is partitioned in time into the current state matrix $X_k$, consisting of the first $N-1$ columns, and the one-step time-shifted state matrix $X_{k+1}$, consisting of the last $N-1$ columns. Correspondingly, let $U_{in} \in \mathbb{R}^{n_u \times (N-1)}$ denote the synchronous input data block aligned with the current state sequence $X_k$. From the current sliding window, the relevant raw inputs are extracted and the input sequence is constructed as
\begin{equation}
	\label{eq:uin}
	U_{in} = [u_1, u_2, \dots, u_{N-1}].
\end{equation}
The corresponding regime-dependent discrete-time system matrices are then estimated as follows
\begin{enumerate}
	\item \textbf{Transition and Observation Matrices ($(\hat{A}_j, \hat{C}_j)$):} To this end, the fundamental shift-invariance property of the extended observability matrix is exploited. Let $\overline{W}_{coher}$ denote the matrix obtained by removing the last $n_y$ block rows of $W_{coher}$, and let $\underline{W}_{coher}$ denote the matrix obtained by removing its first $n_y$ block rows. The transition matrix is then obtained algebraically as
	\begin{equation}
		\hat{A}_j = \overline{W}_{coher}^{\dagger} \, \underline{W}_{coher}.
		\label{eq:hatAj}
	\end{equation}
	Furthermore, the first block row of the observability matrix directly contains the local output mapping. Hence, $\hat{C}_j$ is obtained as the top $n_y \times n$ submatrix of $W_{coher}$
	\begin{equation}
		\hat{C}_j = \begin{bmatrix} I_{n_y} & 0 \end{bmatrix} W_{coher},
		\label{eq:hatCj}
	\end{equation}
	where $I_{n_y}$ is the identity matrix of dimension $n_y$ and $0$ is an appropriately sized zero matrix, isolating the first $n_y$ rows.
	
	\item \textbf{Input Matrix ($\hat{B}_j$):} To mitigate overfitting and attenuate high-energy numerical artifacts induced by uncompensated forced-response residuals, the input matrix is estimated via Tikhonov-regularized least squares \cite{tikhonov1977solutions}:
	\begin{equation}
		\hat{B}_j = (X_{k+1} - \hat{A}_j X_k) U_{in}^T (U_{in} U_{in}^T + \lambda_B I)^{-1},
		\label{eq:hatBj}
	\end{equation}
	where $\lambda_B > 0$ is a strategically selected, small regularization parameter.
\end{enumerate}

The coherently extracted realization $(\hat{A}_j, \hat{B}_j, \hat{C}_j)$ provides an approximation of the local LTI dynamics
of the current operational regime without disrupting the established coordinate frame.

\subsection{Switched Observer and Design Requirements}

Building upon the realization $(\hat{A}_j, \hat{B}_j, \hat{C}_j)$ identified over the active interval $\mathcal{I}_j$, a discrete-time predictor–corrector (filtering) observer is synthesized. In contrast to standard predictive Luenberger architectures, which exhibit an inherent one-step estimation delay, the proposed filtering structure reduces latency by directly incorporating the most recent measurement $y_{k+1}$ into the state update \cite{franklin1998digital}. The observer dynamics are formally defined as
\begin{subequations}
	\begin{align}
		x^p_{k+1} &= \hat{A}_j \hat{x}_k + \hat{B}_j u_k, \label{eq:pred} \\
		\hat{x}_{k+1} &= x^p_{k+1} + L_j (y_{k+1} - \hat{C}_j x^p_{k+1}), \label{eq:corr}
	\end{align}
\end{subequations}
where $x^p_{k+1} \in \mathbb{R}^n$ denotes the \textit{a priori} state prediction, $\hat{x}_{k+1} \in \mathbb{R}^n$ is the \textit{a posteriori} state estimate, and $L_j \in \mathbb{R}^{n \times n_y}$ is the observer gain matrix computed specifically to stabilize the error dynamics during the $j$-th operational regime.

To guarantee the asymptotic convergence of the proposed estimation framework, it is essential to evaluate the nominal error dynamics. Let 
\begin{equation}
	e_k = x_k - \hat{x}_k
	\label{eq012}
\end{equation}
denote the state estimation error. Under nominal conditions, i.e. assuming the coherently extracted matrices $(\hat{A}_j, \hat{B}_j, \hat{C}_j)$ perfectly capture the local linear dynamics and temporarily neglecting the stochastic noise terms, the ideal observation at the next time step is 
\begin{equation}
	y_{k+1} = \hat{C}_j x_{k+1}.
	\label{eq:observation}
\end{equation}
Substituting the \textit{a posteriori} update \eqref{eq:corr} into the error definition at step $k+1$ yields
\begin{equation}
	e_{k+1} = (I - L_j \hat{C}_j)(x_{k+1} - x^p_{k+1}). \label{eq:error_intermediate}
\end{equation}
By expanding the true state transition and subtracting the \textit{a priori} prediction \eqref{eq:pred}, the pre-correction error difference becomes
\begin{equation}
	x_{k+1} - x^p_{k+1} = \hat{A}_j e_k.
\end{equation}
Consequently, by substituting this relation into \eqref{eq:error_intermediate}, the closed-loop nominal error dynamics for the filtering observer are governed by the following autonomous difference equation \cite{simon2006optimal}
\begin{equation}
	e_{k+1} = (I - L_j \hat{C}_j)\hat{A}_j e_k. \label{eq:error_final}
\end{equation}
Therefore, to ensure asymptotic stability over the $j$-th operational interval, the observer gain $L_j$ must be designed such that the matrix $(I - L_j \hat{C}_j)\hat{A}_j$ is Schur stable, i.e., its spectral radius is strictly less than one.

To track the system state under varying operating conditions, the proposed estimator is formulated as a discrete-time switched linear observer. Let $j \in \mathbb{N}$ denote the index of the active operational regime, associated with the discrete-time interval $\mathcal{I}_j = [k_j, k_{j+1}-1]$. The observer dynamics are driven by a piecewise-constant switching signal $\sigma(k)$, such that $\sigma(k) = j$ for all $k \in \mathcal{I}_j$, which selects the active set of system matrices. The observer in \eqref{eq:pred}–\eqref{eq:corr} is then extended to the switched setting as follows
\begin{align}
	x^p_{k} &= \hat{A}_{\sigma(k)} \hat{x}_{k-1} + \hat{B}_{\sigma(k)} u_{k-1}, \label{eq:switched_pred} \\
	\hat{x}_{k} &= x^p_{k} + L_{\sigma(k)} (y_{k} - \hat{C}_{\sigma(k)} x^p_{k}), \label{eq:switched_corr}
\end{align}
where the switching signal $\sigma(k)$ transitions from $j$ to $j+1$ only at the update instant $k_{j+1}$, which is triggered either deterministically by the scheduling period $T_{\rm sched}$ or asynchronously by an anomaly detection event.
Over each operational interval $\mathcal{I}_j$, the signal remains constant, allowing the observer to evolve locally as an LTI system driven by the corresponding gain $L_j$.

As established in the preceding sections, the nonlinear physical process is sequentially approximated across contiguous operational intervals $\mathcal{I}_j$ via a coherent, data-driven piecewise LTI realization $(\hat{A}_j, \hat{B}_j, \hat{C}_j)$. To reconstruct the unmeasurable state trajectories, the discrete-time filtering observer defined by \eqref{eq:switched_pred}-\eqref{eq:switched_corr} has been synthesized. However, in large-scale networked systems, continuously querying the full array of $n_y$ available sensors is often costly in terms 
of computation and energy consumption. Therefore, an effective estimation strategy must not only track the system dynamics but also dynamically select a sparse and informative subset of active sensors \cite{lin2013design}.

Formally, at the onset of each operating interval $\mathcal{I}_j$,   the synthesis problem consists of designing the observer gain matrix $L_j \in \mathbb{R}^{n \times n_y}$ to concurrently satisfy three competing objectives:
\begin{enumerate}
	\item \textbf{State Tracking Performance:} Minimize the \textit{a posteriori} estimation error to ensure accurate reconstruction of the local system trajectories.
	\item \textbf{Switched Stability Requirement:} While it is necessary to enforce the local asymptotic convergence by bounding the spectral radius of the closed-loop transition matrix, i.e.
	\begin{equation}
		\rho\big((I - L_j \hat{C}_j)\hat{A}_j\big) < 1, 
		\label{eq:spectral}
	\end{equation}	
	local Schur stability alone is generally insufficient in switched systems. Therefore, the scheduling strategy  
must enforce conditions ensuring nominal asymptotic stability
and uniform ultimate boundedness  under bounded disturbances and identification mismatches of the overall switched error dynamics \eqref{eq:error_final}. This requires the synthesis of \(L_j\) to be constrained by
stability conditions compatible with the switched-error analysis,
so as to prevent divergent error trajectories during regime
transitions  \cite{liberzon2003switching}.
	\item \textbf{Sensor Sparsity:} Promote column sparsity within the gain matrix $L_j$. Structurally, driving the $i$-th column of $L_j$ to a zero vector implies that the corresponding $y_{k+1}^{(i)}$ measurement channel is completely decoupled from the update equation \eqref{eq:switched_corr}. This is mathematically and operationally equivalent, at the observer level, to deactivating the i-th measurement channel during the current regime. \cite{lin2013design}. 
\end{enumerate}
Finally, the proposed framework must ensure structural resilience by maintaining reliable state estimation during severe transients or hardware faults. In such critical scenarios, it must  prioritize absolute stability and tracking accuracy over sensor sparsity to guarantee fail-safe operation.


Solving this multi-objective framework relies entirely on the continuous stream of real-time operational data.

Consequently,  the basis-aligned realization matrices and the structurally sparse gain $L_j$ are updated at each trigger event for which the available data are sufficiently informative. This requires translating the aforementioned requirements into a mathematically tractable convex optimization routine, ensuring sparse observer updates consistent with the stated
stability requirements.

\section{Adaptive Sparse Observer via Convex Optimization}
\label{sec:opt_formulation}
Large-scale cyber-physical systems may exhibit structural changes and sensor degradation, causing static models to fail. To overcome these limitations, this section introduces a cohesive, adaptive data-driven architecture. In particular, the proposed data-driven framework couples subspace identification with multi-objective convex optimization to dynamically update system realizations, optimize sparse sensor schedules, and ensure stability.


\subsection{Multi-Objective Convex Formulation}
To address the state tracking and sensor sparsity objectives, the synthesis of the observer gain $L_j \in \mathbb{R}^{n \times n_y}$ is recast as a multi-objective convex optimization problem \cite{boyd2004convex}. 
The formulation relies on the synchronous data matrices
obtained from the coherent subspace identification procedure
previously described.

Let $X_k \in \mathbb{R}^{n \times (N-1)}$ and $X_{k+1} \in \mathbb{R}^{n \times (N-1)}$ denote the current and one-step-shifted coherent state sequences, respectively. Furthermore, let $U_{in}$ be the synchronous input data block as defined in \eqref{eq:uin}. To ensure the temporal alignment required by the filtering observer structure \eqref{eq:switched_corr}, the one-step-shifted output measurement matrix $Y_{sync} \in \mathbb{R}^{n_y \times (N-1)}$ is defined as
\begin{equation}
	Y_{sync} = [y_2, y_3, \dots, y_N].
\end{equation}
This matrix contains the real-time measurements used to correct the predictions at each step of the identification batch. Consequently, the \textit{a priori} open-loop state prediction over the entire batch is defined as
\begin{equation}
	X_{pred} = \hat{A}_j X_k + \hat{B}_j U_{in},
\end{equation}
and the baseline output error (innovation sequence) prior to the observer correction as
\begin{equation}
	E_Y = Y_{sync} - \hat{C}_j X_{pred}.
	\label{error}
\end{equation}
Given the operating interval $\mathcal{I}_j$, the optimal gain matrix, denoted as $L_j^*$, is obtained by minimizing a composite cost function $J_j(L)$ that mathematically balances tracking accuracy against structural sparsity
\begin{equation}
	L_j^* = \arg\min_{L \in \mathbb{R}^{n \times n_y}} J_j(L),
\end{equation}
where the objective function is defined as
\begin{equation}
	J_j(L) = \mathcal{J}_{track}(L) + \lambda \mathcal{J}_{sparse}(L) + \gamma \mathcal{J}_{reg}(L). \label{eq:cost_function}
\end{equation}
The term $\mathcal{J}_{track}(L)$ in (\ref{eq:cost_function}) penalizes the \textit{a posteriori} estimation error over the data batch collected during $\mathcal{I}_j$
\begin{equation}
	\mathcal{J}_{track}(L) = \| X_{k+1} - (X_{pred} + L E_Y) \|_F.
\end{equation}
The second term, $\mathcal{J}_{sparse}(L)$, represents the core of the proposed sensor scheduling mechanism. 

To promote the deactivation of redundant sensors while preserving the estimation structure,  an $L_{2,1}$-norm regularization (commonly known as Group Lasso penalty \cite{yuan2006model}) is applied over the columns of the gain matrix
\begin{equation}
	\mathcal{J}_{sparse}(L) = \| L \|_{2,1} = \sum_{i=1}^{n_y} \| L^{(i)} \|_2,
\end{equation}
where $L^{(i)}$ denotes the $i$-th column vector of $L$. Unlike the standard $L_1$-norm, which heuristically promotes unstructured element-wise sparsity, the $L_{2,1}$-norm penalizes the Euclidean norm of entire columns. This regularization promotes group sparsity, encouraging entire columns associated with uninformative measurement channels to become zero. 
If $\|L^{(i)}\|_2=0$, the $i-th$ measurement channel is structurally decoupled from the observer correction step. Hence, the optimization determines an observer-level inactive sensor channel. Physical sensor deactivation or reallocation can subsequently be implemented  \cite{lin2013design}, \cite{joshi2009sensor}.

Finally, the third term, $\mathcal{J}_{reg}(L) = \| L \|_F$, acts as a standard Tikhonov (Ridge) regularization. The Frobenius-norm term penalizes the overall magnitude of the observer gains, helping to prevent numerical ill-conditioning and mitigate the amplification of high-frequency measurement noise. Minimizing the empirical state tracking error $\mathcal{J}_{track}(L)$ over a finite data batch does not, in general, guarantee asymptotic stability of the observer for future out-of-sample trajectories. To enforce stability over the interval $\mathcal{I}_j$, the nominal error dynamics governed by the closed-loop transition matrix $\Phi_j = (I - L \hat{C}_j)\hat{A}_j$ must be asymptotically stable.

While the spectral radius condition $\rho(\Phi_j) < 1$ is non-convex and generally difficult to enforce directly, a sufficient convex condition is introduced by bounding the induced $2$-norm (spectral norm) of the transition matrix \cite{boyd1994linear} 
\begin{equation}
	\| (I - L \hat{C}_j)\hat{A}_j \|_2 \le 1 - \epsilon, \label{eq:hard_constraint}
\end{equation}
where $\epsilon \in (0,1)$ is a small positive stability margin. According to standard matrix analysis \cite{horn2012matrix}, the spectral radius of any square matrix is upper-bounded by its induced 2-norm ($\rho(\Phi_j) \le \| \Phi_j \|_2$). Consequently, satisfying the convex norm constraint \eqref{eq:hard_constraint}, which can be equivalently expressed as an LMI via the Schur complement, provides a sufficient condition to guarantee that the transition matrix $\Phi_j$ is Schur stable.

The complete multi-objective convex optimization problem is formalized as follows
\small
\begin{align}
	L_j^* = \arg\min_{L} \quad & \| X_{k+1} - (X_{pred} + L E_Y) \|_F + \lambda \| L \|_{2,1} + \rho \| L \|_F \label{eq:opt_probl} \\
	\text{s.t.} \quad & \| (I - L \hat{C}_j)\hat{A}_j \|_2 \le 1 - \epsilon \label{eq:cvx_problem} \\
	& L_{:, i} = \mathbf{0}, \quad \forall i \in \mathcal{S}_{excl} \label{eq:excl_constraint}
\end{align}
\normalsize
where $\mathcal{S}_{excl} \subset \{1, \dots, n_y\}$ denotes the set of indices corresponding to sensors that are marked as unavailable during the current interval $\mathcal{I}_j$.

\begin{remark}[Sparsity vs. Stability Trade-off] \label{rem:sparsity}
The \(L_{2,1}\)-norm regularization promotes column sparsity in the observer gain matrix \(L\), and therefore encourages the deactivation of entire measurement channels. However, the resulting sensor set should be interpreted as a sparsity-promoting convex relaxation of the original minimum-cardinality selection problem. Hence, the method does not claim global optimality with respect to the exact number of selected sensors, but rather computes a sparse sensor configuration compatible with the imposed tracking and stability requirements. \demo
\end{remark}

\begin{remark}[Conservatism vs. Structural Sparsity Trade-off] \label{rem:feasibility}
The induced 2-norm constraint in \eqref{eq:hard_constraint} is a sufficient, but generally conservative, condition for Schur stability. Indeed, it enforces a monotone contraction of the
estimation error in the Euclidean norm, whereas Schur stability only requires the spectral radius of the closed-loop matrix to be strictly smaller than one. As a consequence, for highly
coupled or non-normal identified realizations, the sparse optimization problem may become infeasible even when a stabilizing observer gain exists.

This conservatism is accepted in order to preserve convexity with respect to the observer gain \(L\) and to apply the \(L_{2,1}\)-norm penalty directly to the columns of \(L\). This is important for sensor scheduling, since a zero column of \(L\) corresponds directly to an inactive measurement channel. Alternative Lyapunov-based LMI formulations, typically based on the change of variables \(Y=PL\), may reduce conservatism or improve feasibility. However, in that case sparsity would be promoted on \(Y\), not directly on \(L=P^{-1}Y\), and column sparsity of the actual observer gain would no longer be guaranteed.

If the sparse norm-constrained problem is infeasible, the adaptive architecture switches to a dense stabilizing fallback gain, for instance computed through a standard Riccati-based observer design. The stability results of Theorems~1--2 apply to any gain, sparse or dense, satisfying the prescribed contraction condition. Therefore, during fallback operation, the dense gain must either satisfy the same bound 
\[
\|(I-L\hat C_j)\hat A_j\|_2 \leq 1-\epsilon,
\]
or its corresponding contraction factor must be explicitly used in the dwell-time analysis. In this case, stability is prioritized over sparsity, and the sparse sensor schedule is restored once the convex sparse design becomes feasible again. \demo   
\end{remark}

\begin{remark}[Pre-conditioned Relaxation of the Monotonicity Constraint] \label{rem:preconditioning}
	To reduce the conservatism of the unweighted 2-norm constraint (\ref{eq:cvx_problem}) and the risk of primal infeasibility during large transients, an alternative preconditioning approach can be adopted.  In particular, a nominal fully dense stabilizing gain can be computed (e.g., via standard LQ observer design) to derive a valid Lyapunov similarity transformation matrix $T$. Then, the convex optimization \eqref{eq:opt_probl} can be solved by replacing the unweighted constraint with its similarity-transformed counterpart
	\begin{equation}
		\| T (I - L \hat{C}_j)\hat{A}_j T^{-1} \|_2 \le 1 - \epsilon.
		\label{eq101}
	\end{equation}
	Since the transformation $T$ is fixed \textit{a priori} as a constant numerical parameter, the optimization problem remains convex with respect to the decision variable $L$, thereby preserving the exact application of the $L_{2,1}$-norm structural sparsity penalty. However, this formulation is inherently suboptimal, as the static transformation $T$, originally tuned for the dense unconstrained case, imposes an implicit restriction on the sparse solution space, which may limit the achievable tracking performance. \demo
\end{remark}

\subsection{Real-Time Dual-Rate Implementation}
To balance real-time tracking latency with the computational overhead of identification and optimization, the observer employs a dual-rate architecture:
\begin{itemize}
	\item \textbf{Fast Loop}: Operates at every discrete time step $k$ as a recursive predictor-corrector. It uses frozen system matrices and the fixed sparse gain $L_j^*$, processing only the active measurement channels to minimize telemetry bandwidth and latency.
\item \textbf{Slow Loop}: Triggered periodically ($T_{\rm sched}$) or asynchronously by anomaly-driven residual spikes ($\Vert{} E_Y \Vert{}_2 > \tau_{th}$), provided the minimum dwell-time condition ($T_j \geq T_{\min}$) is met. Upon activation, it temporarily reactivates all candidate sensors to extract a full-dimensional recent data window, performs Procrustes-aligned subspace identification, and solves the convex optimization problem to update the local dynamics and synthesize a new fault-aware sensor schedule.
\end{itemize}

The complete execution flow of the proposed framework is formalized in Algorithm \ref{alg:adaptive_sparse_observer}.


\small
\begin{algorithm}[t]
	\caption{Adaptive Data-Driven Sparse Observer}
	\label{alg:adaptive_sparse_observer}
	\begin{algorithmic}[1]
		\Require Data stream $\{u_k,y_k\}$, window $W_{id}$, parameters $\lambda$, $\epsilon$, $T_{\rm sched}$, $T_{\min}$, residual threshold $\tau_{th}$.
		
		\State \textbf{Initialization (Batch Phase):}
		\State Extract initial data batch over $W_{id}$.
		\State Compute initial realization $(\hat A_0,\hat B_0,\hat C_0)$ via Subspace ID.
		\State Store initial observability basis as $W_{old}$.
		\State Solve \eqref{eq:opt_probl} to extract the initial sparse gain $L_0^\ast$.
		\State Initialize regime index $j=0$, transition timer $k_j=W_{id}$, and anomaly flag $\mathcal F_k=0$.
		
		\State \textbf{Online Operation (Fast Loop):}
		\For{each time step $k>W_{id}$}
			\State Acquire input $u_{k-1}$ and measurements $y_k$ from the active sensor subset.
			\State Predict: $x_k^p=\hat A_j\hat x_{k-1}+\hat B_j u_{k-1}$.
			\State Compute innovation: $E_{Y,k}=y_k-\hat C_j x_k^p$.
			\State Correct: $\hat x_k=x_k^p+L_j^\ast E_{Y,k}$.
			
			\State \textbf{Adaptation Trigger (Slow Loop):}
			\If{$(k-k_j)\geq T_{\min}$ \textbf{and} $\big((k-k_j)\geq T_{\rm sched}$ \textbf{or} $\|E_{Y,k}\|_2>\tau_{th}$ \textbf{or} $\mathcal F_k=1\big)$}
				\State Extract recent data window over $[k-W_{id},k]$ and construct Hankel matrices.
				\State Reset anomaly flag: $\mathcal F_k\gets 0$.
				
				\If{$\operatorname{rank}(U_H)<L_h\cdot n_u$}
					\State Skip adaptation: data are not persistently exciting.
					\State Hold current matrices and gain.
				\Else
					\State Compute TSVD to extract the new unaligned basis $W_{new}$.
					\State Solve Procrustes alignment:
					\[
					T_{align}=
					\arg\min_{T^T T=I}
					\|W_{new}T-W_{old}\|_F .
					\]
					\State Update coherent basis: $W_{coher}\gets W_{new}T_{align}$.
					\State Extract updated transition matrices $(\hat A_{j+1},\hat B_{j+1},\hat C_{j+1})$.
					\State Define $\mathcal S_{excl}$ if faulty sensors are isolated.
					\State Solve \eqref{eq:opt_probl} for $L_{j+1}^\ast$; apply fallback to a dense gain if primal infeasible.
					\State Set $W_{old}\gets W_{coher}$, update $j\gets j+1$, and reset timer $k_j\gets k$.
				\EndIf
				
			\ElsIf{$\|E_{Y,k}\|_2>\tau_{th}$ \textbf{and} $(k-k_j)<T_{\min}$}
				\State Store anomaly flag: $\mathcal F_k\gets 1$.
			\EndIf
		\EndFor
	\end{algorithmic}
\end{algorithm}
\normalsize

\section{Properties}
\label{sec:properties}

\begin{theorem}[Intra-interval Stability and Robust Boundedness]\label{thm:stability}
Consider the local LTI approximation associated with the
active operational interval \(I_j\). Let
\[
\Phi_j = (I-L_j^\ast \hat C_j)\hat A_j
\]
be the nominal closed-loop error transition matrix of the
filtering observer. Assume that the gain \(L_j^\ast\) satisfies
the contraction condition
\[
\|\Phi_j\|_2 \leq 1-\epsilon,
\qquad \epsilon\in(0,1).
\]
Then, in the nominal disturbance-free case, the estimation
error satisfies
\[
\|e_k\|_2
\leq
(1-\epsilon)^{k-k_j}\|e_{k_j}\|_2,
\qquad k\in I_j,
\]
and therefore converges exponentially to zero along the
interval.
Moreover, suppose that the effects of process disturbances,
measurement noise, finite-data identification errors, and local
linearization residuals can be collected into an additive
perturbation term \(d_k\) satisfying
\[
\|d_k\|_2 \leq \bar d_j,
\qquad k\in I_j.
\]
Then the perturbed intra-interval error is uniformly ultimately
bounded. In particular,
\begin{equation}\label{eq1}
\|e_k\|_2
\leq
(1-\epsilon)^{k-k_j}\|e_{k_j}\|_2
+
\frac{1-(1-\epsilon)^{k-k_j}}{\epsilon}\bar d_j,
\end{equation}
and hence ultimately
\[
\|e_k\|_2
\leq
\frac{\bar d_j}{\epsilon}
\]
\end{theorem}
Moreover, by assuming that the regime-dependent bounds are uniformly bounded, i.e.,
$\sup_j \bar d_j < \infty$, the perturbed intra-interval estimation error is uniformly ultimately bounded.

\begin{proof}
In the nominal disturbance-free case, the local filtering
observer yields the error dynamics
\[
e_{k+1}
=
(I-L_j^\ast \hat C_j)\hat A_j e_k
=
\Phi_j e_k .
\]
Taking the Euclidean norm and using the induced matrix norm
gives
\[
\|e_{k+1}\|_2
\leq
\|\Phi_j\|_2\|e_k\|_2 .
\]
Since \(\|\Phi_j\|_2\leq 1-\epsilon\), with
\(\epsilon\in(0,1)\), it follows that
\[
\|e_{k+1}\|_2
\leq
(1-\epsilon)\|e_k\|_2 .
\]
By recursive application from \(k_j\) to any \(k\in I_j\), one
obtains
\[
\|e_k\|_2
\leq
(1-\epsilon)^{k-k_j}\|e_{k_j}\|_2 .
\]
Thus, the nominal estimation error converges exponentially to
zero inside the interval.

In the perturbed case, the effects of process disturbances,
measurement noise, finite-data identification errors, and local
linearization residuals are collected into an additive term
\(d_k\), so that
\[
e_{k+1}=\Phi_j e_k+d_k .
\]
Taking norms yields
\[
\|e_{k+1}\|_2
\leq
\|\Phi_j\|_2\|e_k\|_2+\|d_k\|_2 .
\]
Using \(\|\Phi_j\|_2\leq 1-\epsilon\) and
\(\|d_k\|_2\leq \bar d_j\), we obtain
\[
\|e_{k+1}\|_2
\leq
(1-\epsilon)\|e_k\|_2+\bar d_j .
\]
Iterating this scalar comparison inequality gives
\[
\|e_k\|_2
\leq
(1-\epsilon)^{k-k_j}\|e_{k_j}\|_2
+
\sum_{\ell=0}^{k-k_j-1}(1-\epsilon)^\ell \bar d_j .
\]
Since
\[
\sum_{\ell=0}^{k-k_j-1}(1-\epsilon)^\ell
=
\frac{1-(1-\epsilon)^{k-k_j}}{\epsilon},
\]
it follows that
\[
\|e_k\|_2
\leq
(1-\epsilon)^{k-k_j}\|e_{k_j}\|_2
+
\frac{1-(1-\epsilon)^{k-k_j}}{\epsilon}\bar d_j .
\]
for any $k\in I_j$. As the interval length increases, the transient term decays exponentially and the error approaches the ultimate bound 
\[
\|e_k\|_2
\leq
\frac{\bar d_j}{\epsilon}.
\]
Finally, if the regime-dependent bounds are uniformly bounded, i.e., $\sup_j \bar d_j < \infty$, then the boundedness is also uniform with respect to $j$. By defining $\bar d := \sup_j \bar d_j$ and replacing $\bar d_j$ with $\bar d$ in (\ref{eq1}), one has that the perturbed intra-interval estimation error is uniformly ultimately bounded
\end{proof}

\begin{theorem}[Switched stability under bounded transition mismatch] \label{thm:switching}
Consider the nominal switched estimation error dynamics. Assume that, for each active regime $j$, the sparse observer gain $L_j^\ast$ satisfies
\begin{equation*}
\|(I-L_j^\ast \hat C_j)\hat A_j\|_2 \leq \alpha < 1.
\end{equation*}
Moreover, assume that at each switching instant $k_j$, the transition between consecutive identified coordinates satisfies
\begin{equation*}
\|e_{k_j}^{+}\|_2 \leq \bar{\mu}\|e_{k_j}^{-}\|_2 + \bar{\delta},
\end{equation*}
with $\bar{\mu}\geq 1$ and $\bar{\delta}\geq 0$. If the dwell-time $\tau_j$ satisfies $\bar{\mu}\alpha^{\tau_j} < 1$,
then the nominal switched observer error converges asymptotically to zero when $\bar{\delta}=0$. If $\bar{\delta}>0$, the error is uniformly ultimately bounded.
\end{theorem}
\begin{proof}
Let \(k_j\) denote the \(j\)-th switching instant and let
\(\tau_j = k_{j+1}-k_j\) be the corresponding dwell-time.
For each interval \(I_j=[k_j,k_{j+1})\), define the local
closed-loop error transition matrix as
\[
\Phi_j = (I-L_j^\ast \hat C_j)\hat A_j .
\]
By assumption, the observer gain \(L_j^\ast\) satisfies
\[
\|\Phi_j\|_2 \leq \alpha < 1,
\]
uniformly for all active regimes \(j\). Therefore, within the interval \(I_j\), the estimation
error satisfies
\[
\|e_{k+1}\|_2 \leq \alpha \|e_k\|_2,
\qquad k\in I_j .
\]
By recursive application over the whole interval, one obtains
\[
\|e_{k_{j+1}^{-}}\|_2
\leq
\alpha^{\tau_j}\|e_{k_j^{+}}\|_2,
\]
where \(e_{k_j^{+}}\) and \(e_{k_{j+1}^{-}}\) denote,
respectively, the estimation error immediately after the
switch at \(k_j\) and immediately before the switch at
\(k_{j+1}\).

At the switching instant, the observer realization is updated
through the newly identified and aligned state-space
coordinates. Due to finite-data identification errors,
unmodelled nonlinear residuals, and possible imperfect state
reset, the transition between consecutive regimes is assumed
to satisfy the bounded jump inequality
\[
\|e_{k_{j+1}^{+}}\|_2
\leq
\bar\mu \|e_{k_{j+1}^{-}}\|_2 + \bar\delta,
\]
where \(\bar\mu\geq 1\) is a uniform upper bound on the
multiplicative transition amplification and
\(\bar\delta\geq 0\) is a uniform upper bound on the additive
transition mismatch.

Combining the intra-interval contraction with the switching
jump bound yields
\[
\|e_{k_{j+1}^{+}}\|_2
\leq
\bar\mu \alpha^{\tau_j}\|e_{k_j^{+}}\|_2+\bar\delta .
\]
Assume now that the dwell-time satisfies
\[
\bar\mu \alpha^{\tau_j} \leq \beta < 1,
\]
uniformly for all \(j\). Equivalently, it is sufficient that
\[
\tau_j >
-\frac{\ln \bar\mu}{\ln \alpha},
\]
since \(0<\alpha<1\). Then, the sampled error sequence
\[
z_j := \|e_{k_j^{+}}\|_2
\]
satisfies the scalar comparison inequality
\[
z_{j+1} \leq \beta z_j + \bar\delta .
\]
Iterating this recursion gives
\[
z_j
\leq
\beta^j z_0
+
\sum_{\ell=0}^{j-1}\beta^\ell \bar\delta
=
\beta^j z_0
+
\frac{1-\beta^j}{1-\beta}\bar\delta .
\]
Hence,
\[
\limsup_{j\to\infty} z_j
\leq
\frac{\bar\delta}{1-\beta}.
\]
Therefore, the switched estimation error is uniformly
ultimately bounded. In the nominal case, namely when the
transition mismatch is absent and \(\bar\delta=0\), the above
bound reduces to
\[
z_j \leq \beta^j z_0,
\]
which implies
\[
\lim_{j\to\infty} \|e_{k_j^{+}}\|_2 = 0.
\]
Since the error also contracts inside each interval according
to
\[
\|e_k\|_2 \leq \alpha^{k-k_j}\|e_{k_j^{+}}\|_2,
\qquad k\in I_j,
\]
it follows that the complete switched error trajectory
converges asymptotically to zero in the nominal case.
Finally, when bounded process disturbances, measurement
noise, and local linearization residuals are included, their
effect can be collected into an additional bounded additive
term in the error recursion. This preserves the same scalar
comparison structure, with a possibly larger ultimate bound.
Consequently, the perturbed switched observer is uniformly
ultimately bounded.
\end{proof}

\begin{remark}[Role of Orthogonal Procrustes Alignment]\label{rem:arbitrary_switching}
The dwell-time condition in Theorem~2 depends on the
transition amplification factor \(\bar\mu\), which accounts for
possible error growth at switching instants. If consecutive
realizations were related by arbitrary similarity
transformations \(T\in GL(n)\), where \(GL(n)\) denotes the
general linear group of nonsingular matrices, the induced norm
of the transition map could be large or poorly conditioned
\cite{horn2012matrix}.

The Orthogonal Procrustes alignment reduces this effect by
restricting the coordinate update to the orthogonal group
\(O(n)\) \cite{schonemann1966generalized}. In the ideal case
of exact coordinate consistency and perfect state reset, the
transformation is norm-preserving, namely
\[
\|T_{\rm align}e\|_2=\|e\|_2,
\qquad
T_{\rm align}^{\top}T_{\rm align}=I,
\]
so that one may take \(\bar\mu=1\). In practical data-driven
settings, finite-data errors, measurement noise, and unmodelled
nonlinear residuals may still introduce a bounded mismatch,
represented by
\[
\|e_{k_j}^{+}\|_2
\leq
\bar\mu\|e_{k_j}^{-}\|_2+\bar\delta .
\]
Thus, Procrustes alignment should be interpreted not as an
unconditional guarantee of a common Lyapunov function for the
switched error dynamics \cite{liberzon2003switching}, but as a
coordination mechanism that prevents artificial amplification
due to arbitrary coordinate changes and helps keep the
switching mismatch uniformly bounded.\demo
\end{remark}
\begin{figure}[H]
    \centering
        \includegraphics[width=0.45\textwidth]{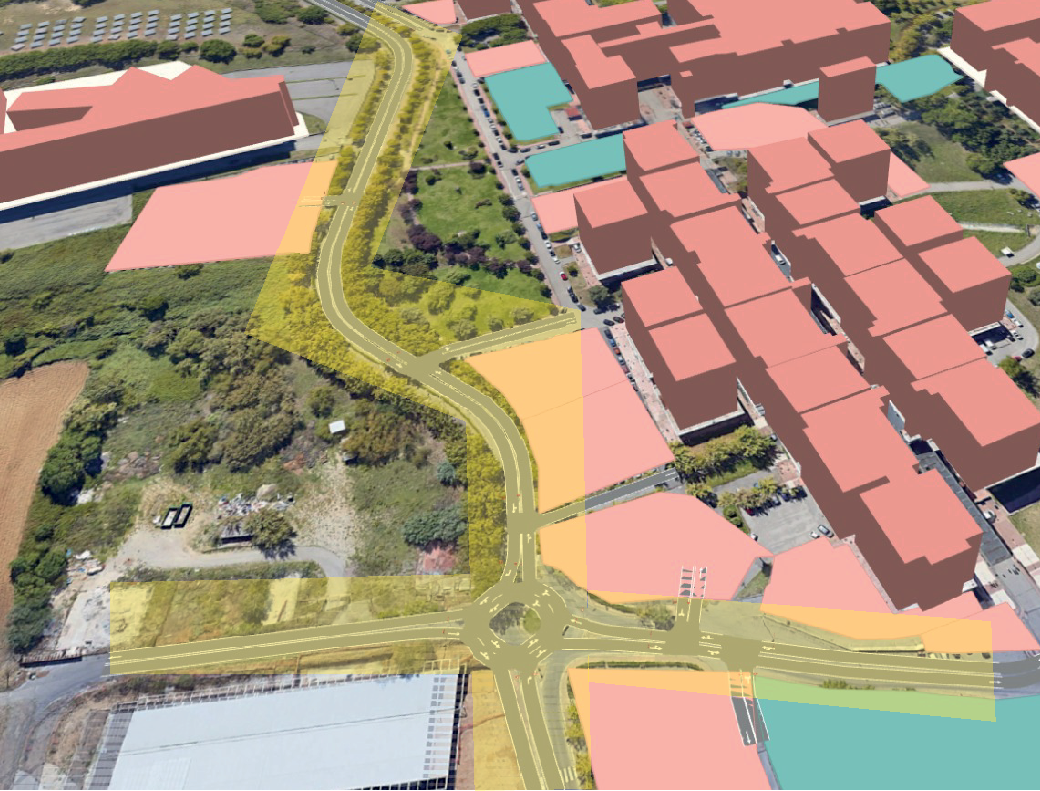}
    \caption{Macroscopic view of a portion of the simulated urban road network, based on the internal campus of the University of Calabria (Italy). The yellow shaded areas explicitly denote a portion of the links comprising the physical plant for the evaluation of the data-driven observer.}
    \label{fig:aimsun}
\end{figure}

\section{Numerical Validation}
\label{sec:simulation}
%
\begin{figure*}[htbp]
    \centering
        \includegraphics[width=0.96\textwidth]{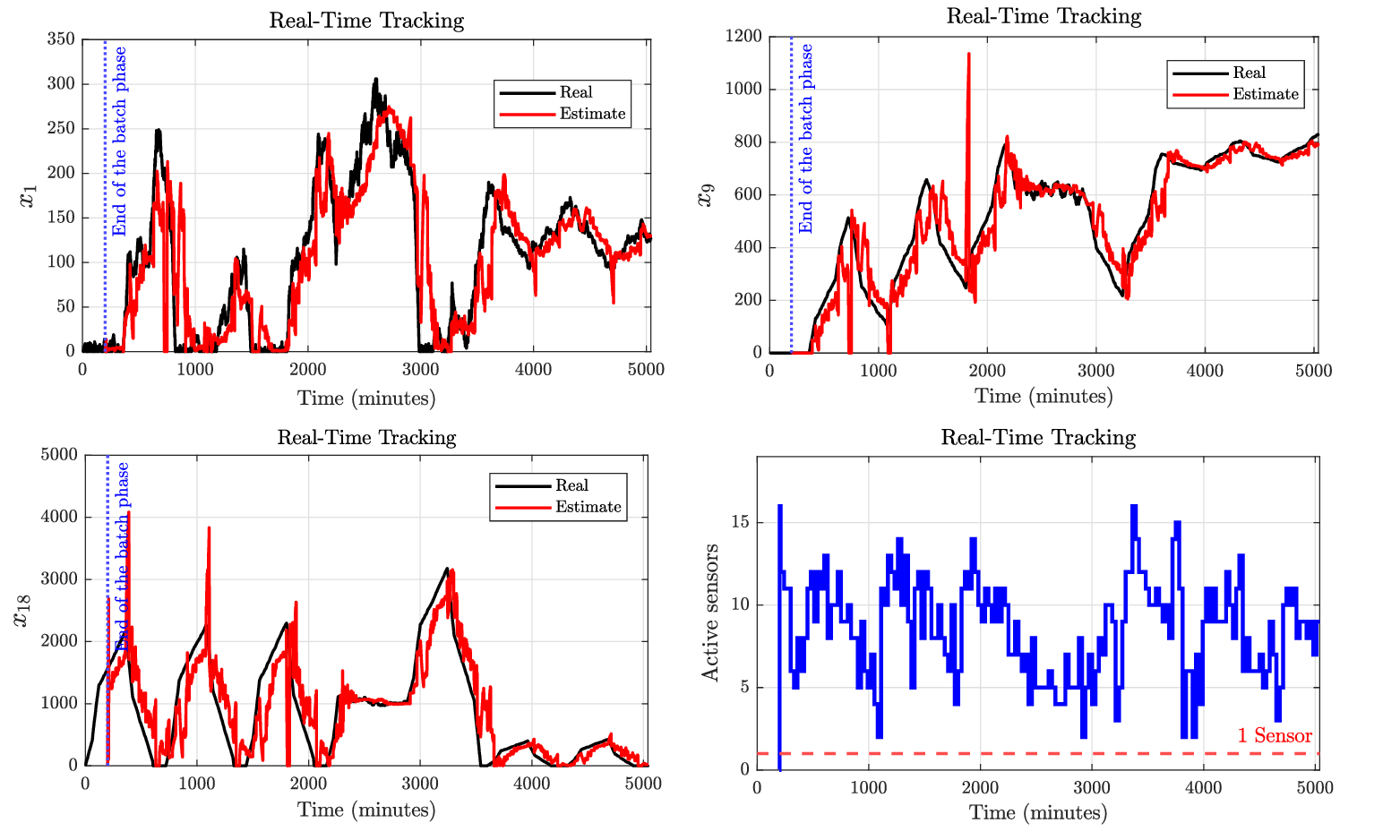}
    \caption{Performance evaluation on the $18$-link network simulated in Aimsun Next. The first three plots show the real-time tracking of representative states ($x_1$, $x_9$, and $x_{18}$), comparing the microscopic ground truth (black) against the observer estimates (red). The bottom-right plot illustrates the dynamic cardinality of the active sensor set synthesized via the $L_{2,1}$-norm optimization.}
    \label{fig:aimsun_results}
\end{figure*}

Validation is performed using the Aimsun Next microscopic simulator \footnote{https://www.aimsun.com/}, where stochastic vehicle behaviors (car-following, lane-changing, gap-acceptance) aggregate into a complex, non-stationary, and unknown nonlinear macroscopic process. The experimental setup is defined as follows. The simulated road network comprises 18 links (Fig. \ref{fig:aimsun}) and 5040-minute simulation horizon is considered. Due to time-varying Origin-Destination (OD) demands the network continuously transitions between fundamentally different dynamic regimes (e.g., peak hours, congestion build-ups, and free-flow traffic). The state vector $x_k \in \mathbb{R}^{18}$ is constructed by averaging vehicle interactions over discrete sampling intervals to extract continuous link-level macroscopic variables. Each candidate sensor provides a noisy, direct measurement of its corresponding link-level macroscopic state. The validation is thus framed as a sparse measurement scheduling problem, dynamically deactivating sensors in links where the local dynamics become highly predictable or redundant.

To explicitly link the theoretical guarantees with the experimental setup, the algorithmic parameters and stability metrics were directly evaluated during the simulation. Specifically, the dwell-time (in this setup $\tau_j = T_{\rm sched} = 30$) and the chosen local contraction margin ($\alpha = 0.98$) yield a decay factor of $\alpha^{\tau_j} = 0.545$. Evaluations across all scheduled transitions confirm that the maximum observed coordinate amplification was $\bar{\mu} \approx 1.774$. Consequently, the switched stability condition (Theorem \ref{thm:switching}) ($\bar{\mu}\alpha^{\tau_j} \approx 0.967 < 1$) is strictly satisfied. Consequently, the sufficient switched-stability condition in Theorem 2 is satisfied, ensuring uniform ultimate boundedness of the switched estimation error in the presence of unmodeled nonlinear residuals ($\bar{\delta} \approx 30.86$).

As illustrated in Fig. \ref{fig:aimsun_results}, the simulation begins with an initial 200-minute phase strictly dedicated to batch data collection. This period is used to extract the initial subspace realization $(\hat{A}_0, \hat{B}_0, \hat{C}_0)$ without applying any online correction. Following this initialization, as formalized in Algorithm \ref{alg:adaptive_sparse_observer}, the local system matrices $(\hat{A}_j, \hat{B}_j, \hat{C}_j)$ are periodically re-identified and updated every $T_{\rm sched}$ to capture the time-varying traffic regimes. Once the predictor-corrector architecture is engaged, the observer tracks the ground-truth macroscopic trajectories for representative states (such as $x_1$, $x_9$, and $x_{18}$), effectively handling the transients inherent to microscopic traffic data.
The most relevant validation of the multi-objective convex formulation is shown in the bottom-right plot of Fig. \ref{fig:aimsun_results}, which illustrates real-time sensor selection. Although full-state direct measurement over the complete network requires $18$ localized sensors, the $L_{2,1}$ group Lasso penalty systematically enforces sparsity by driving uninformative columns of the observer gain $L_j$ to zero. As a result, the framework adapts to evolving traffic conditions by dynamically reducing the active sensor set (typically ranging from $2$ to $15$). This successfully reduces telemetry requirements while maintaining satisfactory tracking performance.
 
To provide a fair performance evaluation, since causal filtering introduces a phase lag during rapid congestion transitions, the Normalized Time-Shifted RMSE (NTS-RMSE) metric is employed instead of the standard RMSE. 
First, the inherent algorithmic latency $\tau_{lag}^{(i)}$ for each state $i$ is identified by maximizing the empirical cross-correlation \cite{shumway2017time}:
\begin{equation}
\tau_{lag}^{(i)} = \arg\max_{\tau} \sum_{k} x^{(i)}_k \hat{x}^{(i)}_{k-\tau}.
\end{equation}
The estimated trajectory is then time-aligned ($\hat{x}^{(i)}_{k+\tau_{lag}^{(i)}}$) over the overlapping window of length $M = N - \tau_{lag}^{(i)}$. To enable fair comparisons across heterogeneous links, the error is normalized by the state's dynamic range $\Delta x^{(i)} = \max_k(x^{(i)}_k) - \min_k(x^{(i)}_k)$, yielding:
\begin{equation}
\text{NTS-RMSE}^{(i)} (\%) = \frac{100}{\Delta x^{(i)}} \sqrt{ \frac{1}{M} \sum_{k=1}^{M} \left( x^{(i)}_k - \hat{x}^{(i)}_{k+\tau_{lag}^{(i)}} \right)^2 }.
\end{equation}


\begin{figure}[tbp]
    \centering
        \includegraphics[width=0.5\textwidth]{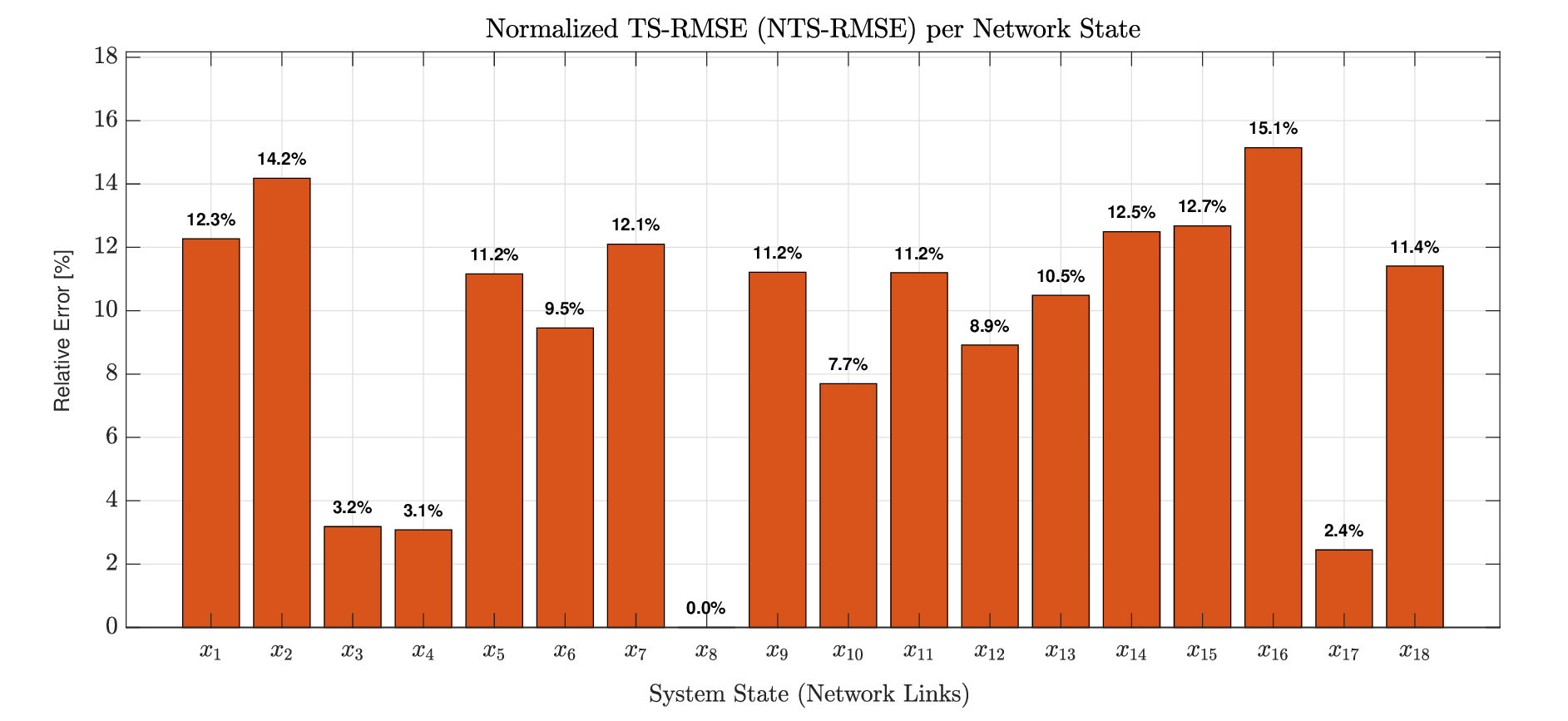}
    \caption{Tracking performance in terms of NTS-RMSE.}
    \label{fig:nts_rmse_bar}
\end{figure}

The NTS-RMSE, evaluated independently for each of  the 18 road links, is reported in Fig. \ref{fig:nts_rmse_bar}. 
The results indicate satisfactory tracking performance despite the significant reduction in active sensors induced by the $L_{2,1}$ penalty. The relative error remains bounded across all spatial links, yielding a network-wide average estimation error of approximately $9.4\%$. While links subject to higher congestion or pronounced transients exhibit increased errors (of up to $15.1\%$ on $x_{16}$ and up to $14.2\%$ on $x_2$), other states are reconstructed with high accuracy (e.g., $2.4\%$ on $x_{17}$ and $3.1\%$ on $x_4$). 
The network-wide average NTS-RMSE of $9.4\%$ indicates that accurate tracking is maintained despite the reduction in active sensors.

Finally, the measured execution times indicate that the proposed dual-rate architecture is computationally compatible with the considered simulation setup. Specifically, the Fast Loop runs in $\approx 0.017\text{ [ms]}$ ($1.74 \times 10^{-5}\text{ [s]}$) per step, introducing negligible computational overhead relative to the traffic sampling interval, while the Slow Loop requires an average of $1.72\text{ [s]}$ per execution to perform data-window extraction, TSVD, Procrustes alignment, and convex optimization.
Because the Slow Loop runs asynchronously as a background supervisory process with an execution time well below the system dwell time ($T_{\rm sched}$), it periodically updates system matrices and sensor schedules without disrupting real-time fast-loop estimation.


\section{Conclusion and Future Work}
\label{sec:conclusion}

This paper presented a data-driven framework for joint state estimation and dynamic sensor scheduling in nonlinear systems. By combining recursive 4SID, Procrustes alignment, and sparse observer synthesis, the proposed approach adapts both the local system representation and the active sensor set to changing operating conditions. Simulation results on an $18$-link Aimsun traffic network showed that accurate state reconstruction can be maintained while substantially reducing the number of active sensors. The dual-rate implementation also demonstrated that the proposed adaptation mechanism can operate without disrupting real-time estimation.
Future work includes extending to distributed multi-agent setups, incorporating event-triggered communication for delays/dropouts, and exploring direct behavioral approaches via Willems' Fundamental Lemma.

\end{document}